\documentclass{ws-ijfcs}
\usepackage{cite}
\usepackage{url}
\usepackage[T1]{fontenc}
\usepackage{hyperref}
\usepackage{amsfonts}
\usepackage{amsmath}
\usepackage{xcolor}
\usepackage{tikz}
\usetikzlibrary{trees}
\usetikzlibrary{positioning}
\usetikzlibrary{matrix}
\usepackage{here}
\usepackage{xcolor}
\newcommand\rev[1]{%
  \protect\leavevmode
  \begingroup
    \color{red}%
    \boldsymbol{#1}%
  \endgroup
}
\usepackage{newtxmath}

\usepackage{color}

\usepackage{tikz}
  \usetikzlibrary{automata,automata, positioning, arrows.meta}
  \usetikzlibrary{arrows}
  \usetikzlibrary{shapes,snakes}
  \usetikzlibrary{decorations.pathmorphing}
  \usetikzlibrary{fit}
  \usetikzlibrary{trees}
  \tikzstyle{every picture}=[
  >=stealth', shorten >=1pt, node distance=1.44cm,auto,bend angle=45,initial text=,
  every state/.style={inner sep=0.75mm, minimum size=1mm},font=\scriptsize,
  triangle/.style = { regular polygon, regular polygon sides=3, inner sep=0}
]

\usepackage{multicol}

\allowdisplaybreaks

\usepackage{enumitem}

\usepackage{musicography}

\begin{document}
\markboth{S.Attou \emph{et al.}}
{Constrained Multi-Tildes: Derived Term and Position Automata}
\title{Constrained Multi-Tildes: Derived Term and Position Automata}
%
%
\author{
  Samira Attou,
  Ludovic Mignot,
  Clément Miklarz and
  Florent Nicart
}
\address{
  GR\textsuperscript{2}IF,
  Université de Rouen Normandie,
  Avenue de l'Université,
  76801 Saint-Étienne-du-Rouvray, France\\
  \email{\{samira.attou2, ludovic.mignot, clement.miklarz1, florent.nicart\}@univ-rouen.fr}
}

\maketitle

\begin{abstract}
  Multi-tildes are regular operators that were introduced to enhance the factorization power of
  regular expressions, allowing us to add the empty word in several factors of a catenation product
  of languages.
  In addition to multi-bars, which dually remove the empty word, they allow representing any acyclic
  automaton by a linear-sized expression, whereas the lower bound is exponential in the classic case.

  In this paper, we extend multi-tildes from disjunctive combinations to any Boolean combination,
  allowing us to exponentially enhance the factorization power of tildes expressions.
  Moreover, we show how to convert these expressions into finite automata and give a Haskell implementation
  of them using advanced techniques of functional programming.

\end{abstract}
%
\section{Introduction}

Regular expressions are widely used inductively defined objects that allow us to easily represent
(potentially infinite) set of words.
In order to solve efficiently the membership test, they can be turned into finite automata~\cite{Ant96}, where the number of states
is linear w.r.t.\ the number of symbols of the expressions.
Numerous operators where added in order to enhance their representation powers, such as Boolean operators.
However, the number of states after the conversion is not necessarily linear anymore~\cite{CCM11b}.

Another class of operators, the multi-tildes~\cite{CCM09}, was introduced in order to
allow a constrained adjunction of the empty word in some factors of the catenation product of regular languages.
In combination with multi-bars~\cite{CCM09b}, multi-tildes allow to improve the factorization power of regular languages:
as an example, it is shown that any acyclic automaton can be turned into a linear-sized equivalent  multi-tildes-bars expression, whereas the lower bound is exponential in the classical case~\cite{CCM12b}.
However, they can be applied only across continuous positions.

In this paper, we extend the idea behind the conception of (disjunctive) multi-tildes to any Boolean combination of them.
These Boolean combinations allow us to extend the specification power of expressions, by \emph{e.g.},
applying tildes across non-continuous intervals of positions.
We show that their actions over languages preserve regularity, that they may lead
to exponentially smaller expressions and how to solve the membership test by defining a finite automaton.

This is the first step of a more general plan: we aim to develop a characterization of the produced automaton in order to inverse the computation, \emph{i.e.}, the conversion of an automaton into a short constrained tildes expression.

The paper is organized as follows.
Section~\ref{sec prelim} contains general preliminaries.
Then, we recall in Section~\ref{sec Bool} classical definitions and constructions for Boolean formulae.
These latter allow us to define constrained tildes in Section~\ref{sec def cons tildes}.
We study their factorization power in Section~\ref{sec fact pow}, and show how to convert these expressions
into finite automata in Section~\ref{sec part der} and Section~\ref{sec Glu}.
Finally, in Section~\ref{sec haskell}, we present a Haskell implementation of these objects.

\section{Preliminaries}\label{sec prelim}
Throughout this paper, we use the following notations:
\begin{itemize}
  \item \( \mathbb{B} \) is the Boolean set \( \{0, 1\} \),
  \item  \(S \rightarrow S'\) is the set of functions from a set \(S\) to a set \(S'\),
  \item for a Boolean \(b\) and a set \(S\), \(S \mid b\) is the set \(S\) if \(b\), \(\emptyset \) otherwise,
  \item \( \subset \) is to be understood as not necessarily strict subset.
\end{itemize}
A \emph{regular expression} \(E\) over an alphabet \(\Sigma \) is inductively defined by
\begin{align*}
  E & = a,         & E & = \emptyset, & E & = \varepsilon, & 
  E & = F \cdot G, & E & = F + G,     & E & = F^*,
\end{align*}
where \(a\) is a symbol in \(\Sigma \) and \(F\) and \(G\) two regular expressions over \(\Sigma \).
Classical priority rules hold: \( ^* > \cdot > + \).
The language \emph{denoted by} \(E\) is the set \(L(E) \) inductively defined by
\begin{align*}
  L(a)        & = \{a\},           & L(\emptyset) & = \emptyset,      & L(\varepsilon) & = \{\varepsilon \}, \\
  L(F\cdot G) & = L(F) \cdot L(G), & L(F+G)       & = L(F) \cup L(G), & L(F^*)         & = {L(F)}^*,
\end{align*}
where \(a\) is a symbol in \(\Sigma \) and \(F\) and \(G\) two regular expressions over \(\Sigma \).
A (non-deterministic) \emph{automaton} \(A\) over an alphabet \(\Sigma \) is a \(5\)-tuple \((\Sigma, Q, I, F, \delta)\) where
\begin{multicols}{2}
  \begin{itemize}
    \item \( Q \) is a finite set of \emph{states},
    \item \(I\subset Q\) is the set of \emph{initial states},
    \item \(F\subset Q\) is the set of \emph{final states},
    \item \(\delta \) is a function in \(\Sigma \times Q \rightarrow 2^Q\).
  \end{itemize}
\end{multicols}
The function \(\delta \) is extended to \(\Sigma \times 2^Q \rightarrow 2^Q\) by
\(\delta(a, P) = \bigcup_{p\in P} \delta(a, p) \) and to \(\Sigma^* \times 2^Q \rightarrow 2^Q\) by
\(\delta(\varepsilon, P)  = P\) and \(\delta(aw, P) = \delta(w, \delta(a, P))\).
The language \emph{denoted} by \(A\) is the set \(L(A) =\{w \in \Sigma^* \mid \delta(w, I) \cap F \neq \emptyset\} \).

Any regular expression with \(n\) symbols can be turned into an equivalent automaton with at most \((n + 1)\) states,
for example by computing the derived term automaton~\cite{Ant96}.
The \emph{partial derivative} of \(E\) w.r.t.\ a symbol \(a\) in \(\Sigma \) is the set of expressions \(\delta_a(E)\)
inductively defined as follows:
\begin{align*}
  \delta_a(b)           & = (\{\varepsilon \} \mid a = b),                                & \delta_a(\emptyset) & = \emptyset,             \\
  \delta_a(\varepsilon) & = \emptyset,                                                    & 
  \delta_a(F+G)         & = \delta_a(F) \cup \delta_a(G),                                                                                  \\
  \delta_a(F \cdot G)   & = \delta_a(F) \odot G \cup (\delta_a(G) \mid \mathrm{Null}(F)), & \delta_a(F^*)       & = \delta_a(F) \odot F^*,
\end{align*}
where \(b\) is a symbol in \(\Sigma \), \(F\) and \(G\) two regular expressions over \(\Sigma \), \(\mathrm{Null}(F) = \varepsilon \in L(F)\) and \(\mathcal{E}\odot G = \bigcup_{E\in \mathcal{E}} \{E \cdot G\} \),
where \( \odot \) has priority over \(\cup\).
The partial derivative of \(E\) w.r.t.\ a word \(w\) is defined by \(\delta_\varepsilon(E) = \{E\} \) and \(\delta_{aw}(E) = \bigcup_{E' \in \delta_a(E)} \delta_w(E') \).
The \emph{derived term automaton} of \(E\) is the automaton \((\Sigma, Q, \{E\}, F, \delta)\) where
\begin{align*}
  Q             & = \bigcup_{w \in \Sigma^*} \delta_w(E), & 
  F             & = \{E' \in Q \mid \mathrm{Null}(E')\},  & 
  \delta(a, E') & = \delta_a(E').
\end{align*}
The derived term automaton of \(E\), with \(n\) symbols, is a finite automaton with at most \((n + 1)\) states
that recognizes \(L(E)\).

A \emph{multi-tilde} is an \(n\)-ary operator which is parameterized by a set \(S\) of couples \((i, j)\) (called \emph{tildes}) in \( {\{1, \ldots, n\}}^2\) with \(i \leq j\). Such an expression is denoted by \(\mathrm{MT}_S(E_1, \ldots, E_n)\) while it is applied over \(n\) expressions \( (E_1, \ldots, E_n) \).
Two tildes \((i, j)\) and \((i', j')\) are \emph{overlapping} if \( \{i, \ldots, j\} \cap \{i', \ldots, j'\} \neq \emptyset \).
A \emph{free subset} of \(S\) is a subset where no tildes overlap each other.
As far as \(S= {(i_k, j_k)}_{k \leq m} \) is free, the action of a tilde is to add the empty word in the catenation of the languages denoted by the expression it overlaps in the catenation of all the denoted languages, \emph{i.e.}
\begin{multline*}
  L(\mathrm{MT}_S(E_1, \ldots, E_n)) =
  L(E_1) \cdot \cdots \cdot L(E_{{i_1}-1}) \cdot (L(E_{i_1}) \cdot \cdots \cdot L(E_{j_1}) \cup \{\varepsilon\}) \cdot L(E_{j_1+1}) \cdot \cdots \\
  \cdots L(E_{i_m - 1}) \cdot  (L(E_{i_m}) \cdot \cdots \cdot L(E_{j_m}) \cup \{\varepsilon\}) \cdot
  L(E_{j_m+ 1} ) \cdots L(E_n).
\end{multline*}
Inductively extended with multi-tildes operators, regular expressions with \(n\) symbols can be turned into equivalent
automata with at most \(n\) states, using the position automaton~\cite{CCM11} or the partial derivation one~\cite{CCM12}.

In the following, we show how to extend the notion of tildes from unions of free subsets to any Boolean combinations of tildes.

\section{Boolean Formulae and Satisfiability}\label{sec Bool}

A \emph{Boolean formula} \(\phi \) over an alphabet \( \Gamma \) is inductively defined by
\begin{align*}
  \phi & = a, & \phi & = o(\phi_1, \ldots, \phi_n),
\end{align*}
where \(a\) is an \emph{atom} in \( \Gamma \), \(o\) is an \emph{operator} associated with an \(n\)-ary function
\( o_f \) from \( \mathbb{B}^n \) to \( \mathbb{B} \),
and \( \phi_1, \ldots, \phi_n\) are \(n\) Boolean formulae over \( \Gamma \).

As an example, \( \neg \) is the operator associated with the Boolean negation,
\( \wedge \) with the Boolean conjunction and
\( \vee \) with the Boolean disjunction.
We denote by \( \bot \) the constant (\(0\)-ary function) \(0\)
and by \( \top \) the constant \(1\).

Let \( \phi \) be a Boolean formula over an alphabet \( \Gamma \).
A function \( i \) from \( \Gamma \) to \( \mathbb{B} \) is said to be an \emph{interpretation} (of \( \Gamma \)).
The \emph{evaluation} of \( \phi \) with respect to \( i \) is the Boolean \(\mathrm{eval}_i(\phi)\) inductively defined by
\begin{align*}
  \mathrm{eval}_i(a) & = i(a), & \mathrm{eval}_i(o(\phi_1, \ldots, \phi_n)) & = o_f(\mathrm{eval}_i(\phi_1), \ldots, \mathrm{eval}_i(\phi_n)),
\end{align*}
where \(a\) an atom in \( \Gamma \), \(o\) is an operator associated with an \(n\)-ary function
\( o_f \) from \( \mathbb{B}^n \) to \( \mathbb{B} \),
and \( \phi_1, \ldots, \phi_n\) are \(n\) Boolean formulae over \( \Gamma \).
Non-classical Boolean functions can also be considered, like in the following example.
\begin{example}\label{ex mirror}
  The operator \( \mathrm{Mirror_n} \) is associated with the \((2 \times n)\)-ary Boolean function \(f\) defined
  for any \((2 \times n)\) Boolean \((b_1, \ldots, b_{2n})\) by
  \begin{align*}
    f(b_1, \ldots, b_{2n}) & \Leftrightarrow (b_1,\ldots, b_n) = (b_{2n}, \ldots, b_{n + 1})
    \Leftrightarrow (b_1 = b_{2n}) \wedge \cdots \wedge (b_n = b_{n + 1})                                                                                                           \\
                           & \Leftrightarrow (b_1 \wedge b_{2n} \vee \neg b_1 \wedge \neg b_{2n}) \wedge \cdots  \wedge (b_n \wedge b_{n + 1} \vee \neg b_n \wedge \neg b_{n + 1}).
  \end{align*}
\end{example}

A Boolean formula is said to be:
\emph{satisfiable} if there exists an interpretation leading to a positive evaluation;
\emph{a tautology} if every interpretation leads to a positive evaluation;
\emph{a contradiction} if it is not satisfiable.

Even if it is an NP-Hard problem~\cite{Coo71}, checking the satisfiability of a Boolean formula can be performed by using incremental algorithms~\cite{DP60,DLL62,Qui82}.
The following method can be performed:
If there is no atom in the formula, then it can be reduced to either \( \bot \) or \( \top \), and it is respectively a tautology or a contradiction;
Otherwise, choose an atom \(a\), replace it with \( \bot \) (denoted by \(a \coloneqq \bot\)), reduce and recursively reapply the method;
If it is not satisfiable, replace \(a\) with \( \top \) (denoted by \(a \coloneqq \top\)), reduce and recursively reapply the method.
The reduction step can be performed by recursively simplifying the subformulae of the form \(o(\phi_1, \ldots, \phi_n)\)
such that there exists \( k\leq n \) satisfying \(F_k \in \{\bot, \top \} \).
As an example, the satisfiability of \( \neg(a \wedge b) \wedge (a \wedge c) \) can be checked as shown in Figure~\ref{fiq ex Quine}.
\begin{figure}[H]
  \begin{center}
    \begin{tikzpicture}[transform shape, scale = 0.8]
      \node (F) [rectangle,draw, rounded corners]  {\( \neg(a \wedge b) \wedge (a \wedge c) \)} ;

      \node (Fabot) [rectangle,draw, rounded corners, below left = of F] {
        \begin{tikzpicture}
          \matrix[matrix of math nodes, nodes={inner sep = 0mm, outer sep = 0mm}, row sep = 1mm] {
            \neg & ( \rev{\bot \wedge b} ) & \wedge       & ( \rev{\bot \wedge c} ) \\
                 & \rev{\neg \bot}         & \wedge       & \bot                    \\
                 & \rev{\top}              & \rev{\wedge} & \rev{\bot}              \\
                 &                         & \bot                                   \\
          };
        \end{tikzpicture}
      };

      \node (Fatop) [rectangle, draw, rounded corners, below right = of F] {
        \begin{tikzpicture}
          \matrix[matrix of math nodes, nodes={inner sep = 0mm, outer sep = 0mm}, row sep = 1mm] {
            \neg & ( \rev{\top \wedge b} ) & \wedge & ( \rev{\top \wedge c} ) \\
                 & \neg b                  & \wedge & c                       \\
          };
        \end{tikzpicture}
      };

      \node (Fbbop) [rectangle, draw, rounded corners, below left = of Fatop] {
        \begin{tikzpicture}
          \matrix[matrix of math nodes, nodes={inner sep = 0mm, outer sep = 0mm}, row sep = 1mm] {
            \rev{\neg \bot} & [1mm] \wedge & [1mm] c \\
            \rev{\top}      & \rev{\wedge} & \rev{c} \\
                            & c                      \\
          };
        \end{tikzpicture}
      };

      \node (Fcbop) [rectangle,draw, rounded corners, below left = of Fbbop] {\( \bot \)} ;
      \node (Fctop) [rectangle,draw, rounded corners, below right = of Fbbop] {\( \top \)} ;

      \draw  (F) -| (Fabot) node [above, near start] {\( a \coloneqq \bot \)};
      \draw (F) -| (Fatop) node [above, near start] {\( a \coloneqq \top \)};
      \draw (Fatop) -| (Fbbop) node [above, near start] {\( b \coloneqq \bot \)};
      \draw (Fbbop) -| (Fcbop)node [above, near start] {\( c \coloneqq \bot \)};
      \draw (Fbbop) -| (Fctop) node [above, near start] {\( c \coloneqq \top \)};

    \end{tikzpicture}
  \end{center}
  \caption{The formula \( \neg(a \wedge b) \wedge (a \wedge c) \) is satisfiable.}%
  \label{fiq ex Quine}
\end{figure}
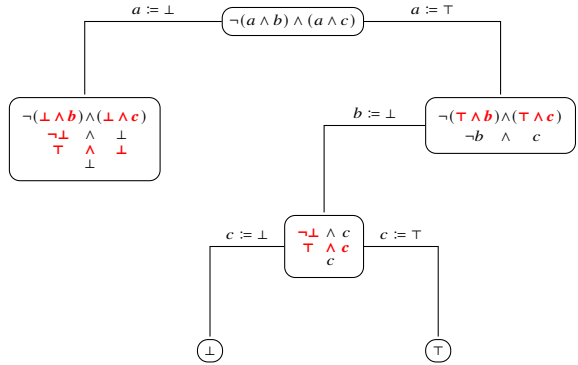

Two Boolean formulae \(\phi\) and \(\phi'\) are \emph{equivalent}, denoted by \(\phi \sim \phi'\),  if for any interpretations \(i\),
\(\mathrm{eval}_i(\phi) = \mathrm{eval}_i(\phi')\).
\begin{example}\label{ex regl simplif mirror}
  Let us consider the operator \(\mathrm{Mirror}_n\) defined in Example~\ref{ex mirror}.
  It can be shown, following the equation in Example~\ref{ex mirror}, that for any \((2n - 1)\) Boolean formulae
  \((\phi_1, \ldots, \phi_{2n - 1})\),
  \begin{align*}
    \mathrm{Mirror}_n(\bot, \phi_1, \ldots, \phi_{2n - 1}) & \sim \mathrm{Mirror}_{n - 1}(\phi_1, \ldots, \phi_{2n - 2}) \wedge \neg \phi_{2n - 1}, \\
    \mathrm{Mirror}_n(\top, \phi_1, \ldots, \phi_{2n - 1}) & \sim \mathrm{Mirror}_{n - 1}(\phi_1, \ldots, \phi_{n - 2}) \wedge \phi_{2n - 1},       \\
    \mathrm{Mirror}_n()                                    & \sim \top.
  \end{align*}
\end{example}

For any two Boolean formulae \(\phi\) and \(\phi'\) and for any atom \(a\), we denote by
\(\phi_{a \coloneqq \phi'}\) the formula obtained by replacing any occurrence of \(a\) in \(\phi\)
with \(\phi'\).
For any two sequences \((\phi'_1, \ldots, \phi'_n)\) of Boolean formulae and \((a_1, \ldots, a_n)\) of distinct atoms, we denote by \(\phi_{a_1 \coloneqq \phi'_1, \ldots, a_n \coloneqq \phi'_n}\) the formula obtained by replacing any occurrence of \(a_k\) in \(\phi\)
with \(\phi'_k\) for any \( 1 \leq k \leq n\).

It is well known that for any Boolean formula \(\phi\) and for any atom \(a\) in \(\phi\),
\begin{equation}\label{eq equiv form}
  \phi \sim \neg a \wedge \phi_{a \coloneqq \bot} \vee a \wedge \phi_{a \coloneqq \top}.
\end{equation}

\section{Constrained Multi-Tildes}\label{sec def cons tildes}

Multi-Tildes operators define languages by computing free sublists of tildes from a set of couples.
This can be viewed as a particular disjunctive combination of these tildes, since sublists of a
free list \(\ell \) define languages that are included in the one \(\ell \) defines.
This disjunctive interpretation can be extended to any Boolean combination.
One may choose to apply conjunctive sequences of not contiguous tildes, or may choose to exclude
some combinations of free tildes.
In this section, we show how to model this interpretation using Boolean formulae.


The action of a tilde is to add the empty word in the catenation of the languages it overhangs.
If the tilde is considered as an interval of contiguous positions \((p_1, p_2, \ldots, p_k)\),
its action can be seen as the conjunction of the substitution of each language at position \( p_1 \),
position \( p_2\), \emph{etc.} with \( \{\varepsilon\} \).

In fact, each position can be considered as an atom of a Boolean formula \( \phi \).
For any interpretation \( i\) leading to a positive evaluation of \( \phi \), we can use \(i(k) \) to
determine whether the language \( L_k \) can be replaced by \( \{\varepsilon\} \) in \( L_1 \cdot \cdots \cdot L_n \).
Let us formalize these thoughts as follows.

Let \( i \) be an interpretation over \( \{1, \ldots, n\} \).
Let \(L_1, \ldots, L_n\) be \( n \) languages.
We denote by \( i(L_1, \ldots, L_n) \) the language \( L'_1 \cdot \cdots \cdot L'_n \) where
\(
L'_k =
\begin{cases}
  \{\varepsilon \} & \text{ if } i(k),  \\
  L_k              & \text{ otherwise.}
\end{cases}
\)

Let \(\phi \) be a Boolean formula over the alphabet \( \{1, \ldots, n\} \) and
\(L_1, \ldots, L_n\) be \(n\) languages.
We denote by \( \phi(L_1, \ldots, L_n) \) the language
\begin{equation}\label{def language}
  \bigcup_{i \mid \mathrm{eval}_i(\phi)}i(L_1, \ldots, L_n).
\end{equation}

\begin{example}\label{ex lang mirror}
  Let us consider the operator \(\mathrm{Mirror}_n\) defined in Example~\ref{ex mirror}
  and the two alphabets \(\Gamma_n = \{1, \ldots, 2n\} \) and \(\Sigma_n = \{a_1, \ldots, a_{2n}\}\).
  Then:
  \begin{multline*}
    \mathrm{Mirror}_n(1, \ldots, 2n)(\{a_1\}, \ldots, \{a_{2n}\}) \\
    = \{w_1 \cdots w_{2n} \mid
    \forall k \leq 2n, w_k \in \{\varepsilon, a_k\}
    \wedge (w_k = \varepsilon \Leftrightarrow w_{2n - k + 1} = \varepsilon) \} \\
    = \{a_1\cdots a_{2_n}, a_1a_3a_4\cdots a_{2n-3}a_{2n-2}a_{2n}, \ldots, a_1a_{2n}, \ldots, a_n a_{n+1}, \varepsilon \}.
  \end{multline*}
\end{example}

First, we remark that the action of constrained tildes preserves regularity,
since it is a finite union of catenations of regular languages, following Equation~\eqref{def language}.
\begin{theorem}
  Let \(\phi \) be a Boolean formula over the alphabet \( \{1, \ldots, n\} \) and
  \(L_1, \ldots, L_n\) be \(n\) regular languages.
  Then
  \(\phi(L_1, \ldots, L_n)\) is regular.
\end{theorem}
Moreover, this definition also allows us to explicit some remarkable identities.
As an example, considering \( n \) languages
\((L_1, \ldots, L_n) \) and a Boolean formula \( \phi \) over the alphabet \( \{1, \ldots, n\} \),
it can be shown that the two following identities hold:
\begin{enumerate}
  \item if \( \phi \) is a contradiction, then
        \(
        \phi (L_1, \ldots, L_n) = \emptyset
        \);
  \item if \( \phi \) is a tautology, then
        \(
        \phi (L_1, \ldots, L_n) = (L_1 \cup \{\varepsilon \}) \cdot \cdots \cdot (L_n \cup \{\varepsilon \})
        \).
\end{enumerate}
Some properties of Boolean formulae can also be transferred while acting over language sequences, as direct consequences of Equation~\eqref{def language}.
\begin{lemma}\label{lem eq form bool}
  Let \(\phi_1 \) and \(\phi_2 \) be two equivalent Boolean formulae over the alphabet \( \{1, \ldots, n\} \) and
  \(L_1, \ldots, L_n\) be \(n\) languages.
  Then
  \(\phi_1(L_1, \ldots, L_n) = \phi_2(L_1, \ldots, L_n)\).
\end{lemma}
\begin{lemma}\label{lem simplif disj}
  Let \(\phi_1 \) and \(\phi_2 \) be two Boolean formulae over \( \{1, \ldots, n\} \) and
  \(L_1, \ldots, L_n\) be \(n\) languages.
  Then
  \((\phi_1 \vee \phi_2)(L_1, \ldots, L_n) = \phi_1(L_1, \ldots, L_n) \cup \phi_2(L_1, \ldots, L_n)\).
\end{lemma}
\begin{lemma}\label{lem simplif conj}
  Let \(\phi \) be a Boolean formula over the alphabet \( \{2, \ldots, n\} \) and
  \(L_1, \ldots, L_n\) be \(n\) languages.
  Then
  \begin{align*}
    (1 \wedge \phi)(L_1, \ldots, L_n)      & = \phi_{2  \coloneqq 1, \ldots, n \coloneqq n - 1}(L_2, \ldots, L_n),           \\
    (\neg 1 \wedge \phi)(L_1, \ldots, L_n) & = L_1 \cdot \phi_{2  \coloneqq 1, \ldots, n \coloneqq n - 1}(L_2, \ldots, L_n).
  \end{align*}
\end{lemma}
As a consequence of Equation~\eqref{eq equiv form}, Lemma~\ref{lem eq form bool}, Lemma~\ref{lem simplif disj} and Lemma~\ref{lem simplif conj}, it holds:
\begin{proposition}\label{prop simplif tildes lang}
  Let \(\phi \) be a Boolean formula over the alphabet \( \{1, \ldots, n\} \) and
  \(L_1, \ldots, L_n\) be \(n\) languages.
  Then
  \begin{align*}
    \phi(L_1, \ldots, L_n) & = L_1 \cdot \phi'(L_2, \ldots, L_n) \cup \phi''(L_2, \ldots, L_n),   \\
    \text{where}\quad
    \phi'                  & = \phi_{1 \coloneqq \bot, 2  \coloneqq 1, \ldots, n \coloneqq n - 1}
    \quad \text{and} \quad
    \phi''                  = \phi_{1 \coloneqq \top, 2  \coloneqq 1, \ldots, n \coloneqq n - 1}.
  \end{align*}
\end{proposition}
\begin{proposition}\label{prop simplif tildes lang2}
  Let \(\phi \) be a Boolean formula over the alphabet \( \{1, \ldots, n\} \) and
  \(L_1, \ldots, L_n\) be \(n\) languages.
  Then
  \begin{align*}
    \phi(L_1, \ldots, L_n) & = \phi'(L_1, \ldots, L_{n-1})\cdot L_n \cup \phi''(L_1, \ldots, L_{n-1}), \\
    \text{where}\quad
    \phi'                  & = \phi_{n \coloneqq \bot}
    \quad \text{and} \quad
    \phi''                  = \phi_{n \coloneqq \top}.
  \end{align*}
\end{proposition}
\begin{example}
  Let us consider the language 
  of Example~\ref{ex lang mirror}:
  \begin{equation*}
    L_n = \mathrm{Mirror}_n(1, \ldots, 2n)(\{a_1\}, \ldots, \{a_{2n}\}).
  \end{equation*}
  Following Proposition~\ref{prop simplif tildes lang}, Proposition~\ref{prop simplif tildes lang2} and rules in Example~\ref{ex regl simplif mirror}, it holds:
  \begin{align*}
    L_n & = \{a_1\} \cdot \mathrm{Mirror}_{n-1}(1, \ldots, 2n-2)(\{a_2\}, \ldots, \{a_{2n-1}\}) \cdot \{a_{2n}\} \\
        & \qquad \cup \mathrm{Mirror}_{n-1}(1, \ldots, 2n-2)(\{a_2\}, \ldots, \{a_{2n-1}\}).
  \end{align*}
\end{example}

This first proposition allows us to show how to easily determine whether the empty word belongs to the action of a Boolean formula over a language sequence and how to compute the quotient of such a computation w.r.t.\ a symbol.
\begin{corollary}\label{cor eps in lang}
  Let \(\phi \) be a Boolean formula over \( \{1, \ldots, n\} \) and
  \(L_1, \ldots, L_n\) be \(n\) languages.
  Then:
  \begin{align*}
    \varepsilon \in \phi(L_1, \ldots, L_n) & \Leftrightarrow   \ \varepsilon \in L_1 \wedge \varepsilon \in \phi'(L_2, \ldots, L_n) \vee \varepsilon \in \phi''(L_2, \ldots, L_n) \\
    \text{where} \quad
    \phi'                                  & = \phi_{1 \coloneqq \bot, 2  \coloneqq 1, \ldots, n \coloneqq n - 1}
    \quad \text{and} \quad
    \phi''                                  = \phi_{1 \coloneqq \top, 2  \coloneqq 1, \ldots, n \coloneqq n - 1}.
  \end{align*}
\end{corollary}

\begin{corollary}\label{cor quot tilde}
  Let \(\phi \) be a Boolean formula over the alphabet \( \{1, \ldots, n\} \),
  \(L_1, \ldots, L_n\) be \(n\) languages and \(a\)  be a symbol.
  Then:
  \begin{align*}
    a^{-1}(\phi(L_1, \ldots, L_n)) & =   \ a^{-1}(L_1) \cdot \phi'(L_2, \ldots, L_n)                        \\
                                   & \qquad \cup (a^{-1}(\phi'(L_2, \ldots, L_n)) \mid \varepsilon \in L_1) \\
                                   & \qquad \cup a^{-1}(\phi''(L_2, \ldots, L_n)),                          \\
    \text{where} \quad
    \phi'                          & = \phi_{1 \coloneqq \bot, 2  \coloneqq 1, \ldots, n \coloneqq n - 1}
    \quad \text{and} \quad
    \phi''  = \phi_{1 \coloneqq \top, 2  \coloneqq 1, \ldots, n \coloneqq n - 1}.
  \end{align*}
\end{corollary}


Let us now extend classical regular expressions with the action of a Boolean formula
considered as a constrained Multi-Tildes.

An \emph{extended to constrained multi-tildes expression} \( E \) over an alphabet \( \Sigma \) (\emph{extended expression} in the following) is inductively defined by
\begin{gather*}
  \begin{aligned}
    E & = a,         & E & = \emptyset,     & E & = \varepsilon, & 
    E & = E_1 + E_2, & E & = E_1 \cdot E_2, & E & = E_1^*,
  \end{aligned}\\
  E = \phi(E_1, \ldots, E_n),
\end{gather*}
where \(a\) is a symbol in \( \Sigma \), \(\phi \) is a Boolean formula over the alphabet
\( \{1, \ldots, n\} \) and \( E_1, \ldots, E_n \) are \( n \) extended expressions over \( \Sigma \).
The \emph{language denoted} by an extended expression \(E\) is the language \( L(E) \) inductively defined by
\begin{gather*}
  \begin{aligned}
    L(a)         & = \{a\},              & L(\emptyset)     & = \emptyset,           & L(\varepsilon) & = \{ \varepsilon \}, \\
    L(E_1 + E_2) & = L(E_1) \cup L(E_2), & L(E_1 \cdot E_2) & = L(E_1) \cdot L(E_2), & L(E_1^*)       & = {L(E_1)}^*,
  \end{aligned}\\
  L(\phi(E_1, \ldots, E_n)) = \phi(L(E_1), \ldots, L(E_n)),
\end{gather*}
where \(a\) is a symbol in \( \Sigma \), \( \phi \) is a Boolean formula over the alphabet
\( \{1, \ldots, n\} \) and \( E_1, \ldots, E_n \) are \( n \) extended expressions over \( \Sigma \).

Since the Boolean satisfiability is an NP-hard problem~\cite{Coo71}, so is the emptiness problem for extended expressions,
as a direct consequence of Equation~\eqref{def language} and of denoted language definition.
\begin{proposition}
  Let \(\Sigma_k = \{a_1, \ldots, a_k\} \) be an alphabet and \(\phi\) be a Boolean formula over the alphabet
  \( \{1, \ldots, k\} \). Then
  \( L(\phi(a_1, \ldots, a_k)) \neq \emptyset \)
  \(\Longleftrightarrow \)
  \( \phi \) is satisfiable.
\end{proposition}
%

\begin{corollary}
  Determining whether the language denoted by an extended expression is empty is NP-hard.
\end{corollary}


\section{Factorization Power}\label{sec fact pow}

In this section, we exhibit a parameterized family of expressions \(E_n\) such that
the smallest NFA recognizing \( L(E_n) \) admits a number of states exponentially larger
than the sum of the number of symbols, the number of atoms and the number of operators of \(E_n\).
Let us consider the alphabet \( \Sigma_{2n} = \{a_1, \ldots, a_{2n} \} \) and
the expression \(E_n = \mathrm{Mirror}_n(1, \ldots, 2n)(a_1, \ldots, a_{2n}) \).
The expression \(E_n\) contains \(2n\) atoms, \(2n\) symbols and \(1\) operator.
Using classical Boolean operators, like \(\wedge \), \(\vee \) and \(\neg \),
the Boolean formula \(\mathrm{Mirror}_n(1, \ldots, 2n)\) can be turned into the equivalent one
\((1 \wedge 2n \vee \neg 1 \wedge \neg 2n) \wedge \cdots  \wedge (n \wedge (n + 1) \vee \neg n \wedge \neg (n + 1))\)
following equation in Example~\ref{ex mirror}, that contains \(4n\) atoms and \((6n - 1)\) operators,
which is a linearly larger Boolean formula.
In order to exhibit a lower bound of the number of states of any NFA recognizing \( L(E_n) \),
let us consider the following property~\cite{GS96}:
\begin{theorem}[\cite{GS96}]\label{thm shallit}
  Let \(L\subset \Sigma^*\) be a regular language, and suppose there exists a set of pairs
  \( P = \{ (x_i , w_i ) : 1 \leq i \leq n\} \) such that
  \(x_i w_i \in L\) for \(1 \leq i \leq n\)
  and
  \(x_i w_j \notin L\) for \(1 \leq i,j \leq n\) and \(i \neq j\).
  Then any NFA accepting \(L\) has at least \(n\) states.
\end{theorem}

For any sequences of \(n\) Booleans \(bs = (b_1, \ldots, b_n)\),
let us consider the words \(v_{bs} = w_1 \cdots w_n\) and \(v'_{bs} = w'_1 \cdots w'_{n}\)
where
\begin{align}
  w_k  & =
  \begin{cases}
    a_k         & \text{ if } \neg b_k, \\
    \varepsilon & \text{ otherwise},
  \end{cases}
       & 
  w'_k & =
  \begin{cases}
    a_{n + k}   & \text{ if } \neg b_k, \\
    \varepsilon & \text{ otherwise}.
  \end{cases}
\end{align}

Denoting by \(\mathrm{rev}(b_1, \ldots, b_n)\) the sequence \( (b_n, \ldots, b_1) \),
since the only words in \( L(E_n) \) are the words \(v_{bs} \cdot v'_{\mathrm{rev}(bs)}\),
and since the words \(v_{bs} \cdot v'_{bs'}\) for any \(bs' \neq \mathrm{rev}(bs)\) are not
in \(L(E_n)\), it holds according Theorem~\ref{thm shallit} that

\begin{proposition}
  There is at least \(2^n\) states in any automaton recognizing \(L(E_n)\).
\end{proposition}
%
\begin{theorem}
  There exist extended regular expressions exponentially smaller than any automaton recognizing their denoted languages.
\end{theorem}

\section{Partial Derivatives and Automaton Computation}\label{sec part der}

Let us now show how to extend the Antimirov method in order to syntactically solve the membership
test and to compute a finite automaton recognizing the language denoted by an extended expression.
%
%
First, we define the partial derivative of an expression w.r.t.\ a symbol,
where the derivation formula for the action of a Boolean combination is obtained
by considering the fact that the empty word may appear at the first position for two reasons:
if the first operand is nullable, or because the empty word is inserted by the multi-tilde.

\begin{definition}
  Let \(E\) be an extended expression and \(a\) be a symbol.
  The \emph{partial derivative} of \(E\) w.r.t.\ \(a\) is the set \(\delta_a(E)\)
  of extended expressions inductively defined as follows:
  \begin{gather*}
    \begin{aligned}
      \delta_a(b)             & = \{\varepsilon\} \mid b = a,
                              & 
      \delta_a(\varepsilon)   & = \emptyset,
      \\
      \delta_a(\emptyset)     & = \emptyset,
                              & 
      \delta_a(E_1 + E_2)     & = \delta_a(E_1) \cup \delta_a(E_2),
      \\
      \delta_a(E_1 \cdot E_2) & = \delta_a(E_1) \odot E_2 \cup \delta_a(E_2) \mid \varepsilon \in L(E_1),
                              & 
      \delta_a(E_1^*)         & = \delta_a(E_1) \odot E_1^*,
    \end{aligned}\\
    \begin{aligned}
      \delta_a(\phi(E_1, \ldots, E_n)) = & \ \delta_a(E_1) \odot \phi'(E_2, \ldots, E_n)                             \\
                                         & \qquad \cup \delta_a(\phi'(E_2, \ldots, E_n)) \mid \varepsilon \in L(E_1) \\
                                         & \qquad \cup \delta_a(\phi''(E_2, \ldots, E_n)),
    \end{aligned}
  \end{gather*}
  where \(b\) is a symbol in \( \Sigma \), \( \phi \) is a Boolean formula over the alphabet
  \( \{1, \ldots, n\} \), \( E_1, \ldots, E_n \) are \( n \) extended expressions over \( \Sigma \)
  and
  \begin{align*}
    \phi'  & = \phi_{1 \coloneqq \bot, 2  \coloneqq 1, \ldots, n \coloneqq n - 1}, & 
    \phi'' & = \phi_{1 \coloneqq \top, 2  \coloneqq 1, \ldots, n \coloneqq n - 1}.
  \end{align*}
\end{definition}

In the following, to shorten the expressions in the next examples, we consider the trivial quotients
\(E \cdot \varepsilon  = \varepsilon \cdot E = E \) and
\(E \cdot \emptyset   = \emptyset \cdot E = \emptyset\).
Furthermore, when \(\phi \) is a contradiction, we consider that \(\phi(E_1, \ldots, E_n) = \emptyset\).
\begin{example}\label{ex calc deriv symb}
  Let us consider the alphabet \(\Sigma = \{a,b\} \) and the expression
  \(E = \mathrm{Mirror}_2(1, 2, 3, 4)(a^+, b^+, a^+, b^+)\), where \(x^+ = x \cdot x^*\).
  The derived terms of \(E\) w.r.t.\ the symbols in \(\Sigma \) are the following, where underlined computations
  equal \(\emptyset\):
  \begin{align*}
    \delta_a(E) & = \delta_a(a^+) \odot \mathrm{Mirror}_2(\bot, 1, 2, 3)(b^+, a^+, b^+) \cup \underline{\delta_a(\mathrm{Mirror}_2(\top, 1, 2, 3)(b^+, a^+, b^+) )}   \\
                & = \{a^*\} \odot (\mathrm{Mirror}_1(1, 2) \wedge \neg 3)(b^+, a^+, b^+)                                                                              \\
                & = \{a^* \cdot (\mathrm{Mirror}_1(1, 2) \wedge \neg 3)(b^+, a^+, b^+)\},                                                                             \\
    \delta_b(E) & = \underline{\delta_b(a^+) \odot \mathrm{Mirror}_2(\bot, 1, 2, 3)(b^+, a^+, b^+)} \cup \delta_b(\mathrm{Mirror}_2(\top, 1, 2, 3)(b^+, a^+, b^+) )   \\
                & = \delta_b((\mathrm{Mirror}_1(1, 2) \wedge 3)(b^+, a^+, b^+) )                                                                                      \\
                & = \delta_b(b^+ ) \odot (\mathrm{Mirror}_1(\bot, 1) \wedge  2)(a^+, b^+) \cup \underline{\delta_b((\mathrm{Mirror}_1(\top, 1) \wedge  2)(a^+, b^+))} \\
                & = \{b^*\} \odot (\neg 1 \wedge  2)(a^+, b^+)
    = \{b^* \cdot (\neg 1 \wedge  2)(a^+, b^+)\}.
  \end{align*}
\end{example}
As usual, the partial derivative is closely related to the computation of the quotient
of the denoted language, as a direct consequence of Corollary~\ref{cor quot tilde},
and by induction over the structure of \(E\).
\begin{proposition}\label{prop deriv symbol quot}
  Let \(E\) be an extended expression and \(a\) be a symbol.
  Then
  \begin{equation*}
    \bigcup_{E' \in \delta_a(E)} L(E') = a^{-1}(L(E)).
  \end{equation*}
\end{proposition}
The partial derivative can be classically extended from symbols to words by repeated applications.
Let \(E\) be an extended expression, \(a\) be a symbol and \(w\) be a word.
Then
\begin{align*}
  \delta_\varepsilon(E) & = \{E\}, & \delta_{aw}(E) & = \bigcup_{E' \in \delta_a(E)} \delta_w(E').
\end{align*}
\begin{example}
  Let us consider the expression \(E\) and its derived terms computed in Example~\ref{ex calc deriv symb}.
  Then:
  \begin{align*}
    \delta_{aa}(E) & = \delta_a(a^* \cdot (\mathrm{Mirror}_1(1, 2) \wedge \neg 3)(b^+, a^+, b^+))
    = \{a^* \cdot (\mathrm{Mirror}_1(1, 2) \wedge \neg 3)(b^+, a^+, b^+)\},                        \\
    \delta_{ab}(E) & = \delta_b(a^* \cdot (\mathrm{Mirror}_1(1, 2) \wedge \neg 3)(b^+, a^+, b^+))  \\
                   & = \delta_b((\mathrm{Mirror}_1(1, 2) \wedge \neg 3)(b^+, a^+, b^+))            \\
                   & = \delta_b(b^+) \odot (\mathrm{Mirror}_1(\bot, 1) \wedge \neg 2)(a^+, b^+)
    \cup \delta_b((\mathrm{Mirror}_1(\top, 1) \wedge \neg 2)(a^+, b^+))                            \\
                   & = \{b^*\} \odot (\neg 1 \wedge \neg 2)(a^+, b^+)
    \cup \delta_b((1 \wedge \neg 2)(a^+, b^+))                                                     \\
                   & = \{b^* \cdot (\neg 1 \wedge \neg 2)(a^+, b^+)\} \cup \delta_b((\neg 1)(b^+)) \\
                   & = \{b^* \cdot (\neg 1 \wedge \neg 2)(a^+, b^+), b^*\}.
  \end{align*}
\end{example}
Once again, this operation is a syntactical representation of the quotient computation,
as a direct consequence of Proposition~\ref{prop deriv symbol quot}, and by induction over the structure of words.
\begin{proposition}\label{lang deriv part word}
  Let \(E\) be an extended expression and \(w\) be a word.
  Then
  \begin{equation*}
    \bigcup_{E' \in \delta_w(E)} L(E') = w^{-1}(L(E)).
  \end{equation*}
\end{proposition}

As a direct consequence, the membership test is solved for extended expressions.
Indeed, determining whether a word \(w\) belongs to the language denoted by an extended expression
\(E\) can be performed by computing the partial derivative of \(E\) w.r.t.\ \(w\) and then by
testing whether it contains a nullable expression, \emph{i.e.} an expression whose denoted language
contains the empty word.

Let us now show that the partial derivative automaton of an extended expression \(E\) is a finite one that recognizes
\(L(E)\).


In the following, we denote by \(\mathcal{D}_E\) the set of derived terms of an expression \(E\), \emph{i.e.}, the set of expressions \(\displaystyle \bigcup_{w\in\Sigma^*} \delta_w(E)\).

Moreover, given an expression \(\phi(E_1, \ldots, E_n)\), an integer \( 1 \leq k \leq n - 1\) and an interpretation \(i\) in \( \{ 1, \ldots, k\} \rightarrow \mathbb{B}\),
we denote by \( \mathcal{D}_{E, k, i} \) the set \(\mathcal{D}_{E_k} \odot \phi'(E_{k+1}, \ldots, E_n) \), where
\begin{equation*}
  \phi' = \phi_{
    1 \coloneqq
    \begin{cases}
      \top & \text{if } i(1),  \\
      \bot & \text{otherwise},
    \end{cases}
    \ldots,
    k \coloneqq
    \begin{cases}
      \top & \text{if } i(k),  \\
      \bot & \text{otherwise},
    \end{cases}
    k+1 \coloneqq 1,
    \ldots,
    n \coloneqq n - k
  }.
\end{equation*}
First, the union of these sets includes  the partial derivatives and is stable w.r.t.\ derivation by a symbol,
by induction over the structures of extended expressions, of words and over the integers.
\begin{proposition}
  Let \(E\) be an extended expression and \(a\) be a symbol.
  Then the two following conditions hold:
  \begin{enumerate}
    \item \( \displaystyle \delta_a(E) \subset \bigcup_{\substack{1 \leq k \leq n, \\ i \in \{ 1, \ldots, k\} \rightarrow \mathbb{B}}} \mathcal{D}_{E, k, i}\),
    \item \( \displaystyle \bigcup_{E' \in \mathcal{D}_{E, k, i} } \delta_a(E') \subset
          \displaystyle \bigcup_{
            \substack{k \leq k' \leq n, \\ i' \in \{ 1, \ldots, k'\} \rightarrow \mathbb{B}}
          } \mathcal{D}_{E, k', i'}
          \).
  \end{enumerate}
\end{proposition}
As a direct consequence, the set of the derived terms of an extended expression is included in the union
of the \( \mathcal{D}_{E, k, i} \) sets.
\begin{corollary}
  Let \(E\) be an extended expression.
  Then
  \begin{equation*}
    \mathcal{D}_E \subset \bigcup_{\substack{1 \leq k \leq n,\\ i \in \{ 1, \ldots, k\} \rightarrow \mathbb{B}}} \mathcal{D}_{E, k, i}.
  \end{equation*}
\end{corollary}


According to a trivial inductive reasoning, one can show that such a set is finite.
\begin{corollary}\label{cor finite set}
  Let \(E\) be an extended expression and \(w\) be a word.
  Then
  \begin{center}
    \(\displaystyle \bigcup_{w \in \Sigma^*} \delta_w(E)\) is a finite set.
  \end{center}
\end{corollary}
As a direct consequence, the derived term automaton of an extended expression, defined as usual
with derived terms as states and transitions computed from partial derivation, fulfils finiteness and correction.
\begin{theorem}
  Let \(E\) be an extended expression and \(a\) be a symbol.
  The partial derivative automaton of \( E \) is a finite automaton recognizing \(L(E)\).
\end{theorem}

\begin{example}
  Let us consider the expression \(E\) defined in Example~\ref{ex calc deriv symb}.
  The derived term automaton of \(E\) is given in Figure~\ref{fig aut ant mirror}.
\end{example}

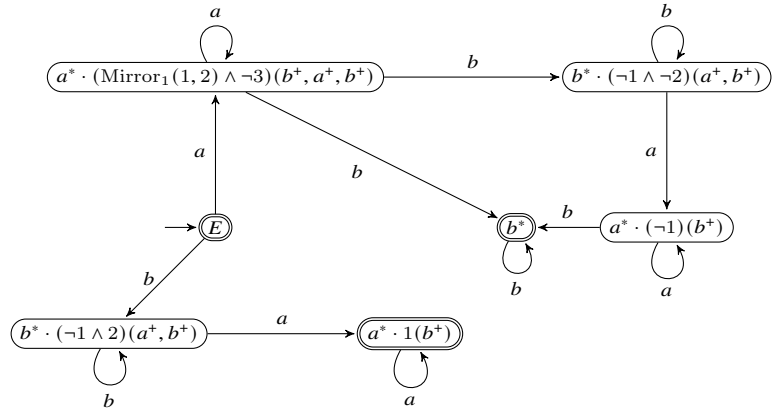
\begin{figure}[H]
  \centering
  \begin{tikzpicture}[node distance=3cm,bend angle=30]
    \node[state, rounded rectangle, accepting, initial left, initial text={}] (1)  {\(E\)} ;
    \node[state, rounded rectangle, below left of=1,node distance = 2cm] (2) {\( b^* \cdot (\neg 1 \wedge  2)(a^+, b^+) \)};
    \node[state, rounded rectangle, right of=2,node distance = 4cm, accepting] (3) {\(a^*\cdot 1(b^+) \)};

    \node[state, rounded rectangle, above of=1,node distance = 2cm] (4) {\( a^* \cdot (\mathrm{Mirror}_1(1, 2) \wedge \neg 3)(b^+, a^+, b^+) \)};
    \node[state, rounded rectangle, right of=4,node distance = 6cm] (5) {\( b^* \cdot (\neg 1 \wedge \neg 2)(a^+, b^+) \)};
    \node[state, rounded rectangle, below of=5,node distance = 2cm] (6) {\( a^* \cdot (\neg 1)(b^+) \)};
    \node[state, rounded rectangle, left of=6,node distance = 2cm, accepting ] (7) {\( b^* \)};

    \path[->]
    (1) edge[ left] node {\( a \)} (4)
    (4) edge[out = 120, in = 60, loop, above] node {\( a \)} ()
    (4) edge[ above] node {\( b \)} (5)
    (5) edge[out = 120, in = 60, loop, above] node {\( b \)} ()
    (5) edge[ left] node {\( a \)} (6)
    (6) edge[out = -120, in = -60, loop, below] node {\( a \)} ()
    (6) edge[ above] node {\( b \)} (7)
    (4) edge[ below left] node {\( b \)} (7)
    (7) edge[out = -120, in = -60, loop, below] node {\( b \)} ()

    (1) edge[ left] node {\( b \)} (2)
    (2) edge[out = -120, in = -60, loop, below] node {\( b \)} ()
    (2) edge[ above] node {\( a \)} (3)
    (3) edge[out = -120, in = -60, loop, below] node {\( a \)} ()
    ;
  \end{tikzpicture}
  \caption{The derived term automaton of \(E\).}%
  \label{fig aut ant mirror}
\end{figure}

\section{The Glushkov Automaton of an Expression}\label{sec Glu}

The Glushkov automaton~\cite{Glu61} is a convenient way to produce an \((n + 1)\)-state
automaton from a \(n\)-width regular expression. Naively, this automaton can be computed in
\(O(n^3)\) time, but this complexity can be reduced to \(O(n^2)\)~\cite{Bru93, PZC96}.

This complexity is also the best known one for the computation of the derived term automaton using
optimized techniques~\cite{CZ01}, reduced from the naive case in \(O(n^5)\)~\cite{CZ01}.

If the derivation technique allows us to solve the membership test without computing the whole
derived term automaton, it is not the case for the Glushkov automaton. Its whole structure is computed
inductively by five functions, \(\mathrm{Pos}\), \(\mathrm{First}\), \(\mathrm{Last}\), \(\mathrm{Follow}\)
and \(\mathrm{Null}\). However, its structure has been deeply studied and characterized~\cite{CZ00}.

Consequently, the extension of the Glushkov method to constrained tildes and the characterization of its structure
is a preliminary step to study the conversion from automata to extended expressions.

\subsection{The Computation for (Classical) Regular Expressions}\label{subsec clas cons}

As recalled before, the computation is based on five functions called \emph{position functions},
defined for any regular expression \(E\) as follows:
\begin{align*}
  \mathrm{Pos}(E)    & = \Sigma_E,                                                              \\
  \mathrm{First}(E)  & = \{a\in\mathrm{Pos}(E) \mid \exists w, aw \in L(E)\},                   \\
  \mathrm{Last}(E)   & = \{a\in\mathrm{Pos}(E) \mid \exists w, wa \in L(E)\},                   \\
  \mathrm{Follow}(E) & = \{(a, b)\in{\mathrm{Pos}(E)}^2 \mid \exists (w, w'), wabw' \in L(E)\}, \\
  \mathrm{Null}(E)   & = \varepsilon \in L(E),
\end{align*}
where \(\Sigma_E\) is the set of symbols that appears in \(E\).

As far as there is no occurrence of \( \emptyset \) in \(E\), these functions can be inductively
computed over the structure of \(E\). As an example, the \(\mathrm{Follow}\) function
is inductively computed as follows:
\begin{gather*}
  \mathrm{Follow}(a) = \mathrm{Follow}(\varepsilon) = \emptyset,\\
  \begin{aligned}
    \mathrm{Follow}(E + F)     & = \mathrm{Follow}(E) \cup \mathrm{Follow}(F),                        \\
    \mathrm{Follow}(E \cdot F) & = \mathrm{Follow}(E) \cup \mathrm{Follow}(F)
    \cup \mathrm{Last}(E) \times \mathrm{First}(F),                                                   \\
    \mathrm{Follow}(E^*)       & = \mathrm{Follow}(E) \cup \mathrm{Last}(E) \times \mathrm{First}(E).
  \end{aligned}
\end{gather*}
Once computed, these functions lead to the definition of the Glushkov automaton.
\begin{definition}\label{def glu aut}
  The \emph{Glushkov automaton} of a regular expression \(E\) is the automaton
  \(G_E=(\mathrm{Pos}(E), \mathrm{Pos}(E) \uplus \{0\} , \{0\}, F, \delta) \)
  defined by
  \begin{align*}
    F            & = \mathrm{Last}(E) \cup \{0\} \mid \mathrm{Null}(E), \\
    \delta(a, 0) & = \{a\} \mid a \in \mathrm{First}(E),                \\
    \delta(a, p) & = \{b\} \mid (a, b) \in \mathrm{Follow}(E),
  \end{align*}
  where \(0\) is not in \(\mathrm{Pos}(E)\) and \(p\) is any state distinct from \(0\).
\end{definition}

However, this automaton does not necessarily recognize \(L(E)\).
Indeed, if a symbol appears twice or more in \(E\), the occurrences
may accept distinct following symbols in the denoted language, or
can be in a last position or not.

As an example, let us consider the expression \(E={(a+b)}^*a(a+b)\).
The first occurrence of the symbol \(b\) makes it belonging to \(\mathrm{First}(E)\),
the second one makes it belonging to \(\mathrm{Last}(E) \).
Therefore, by construction, \(b\) is in \(L(G_E)\), but not in \(L(E)\).

A sufficient condition is when \(E\) is \emph{linear}, \emph{i.e.} when
any symbol appears only once in \(E\).
In this case, \(L(E) = L(G_E)\). The position functions do not mix the data
obtained from several positions, since any symbol appear only at one position.

If \(E\) is not linear, the position automaton is produced as follows:
\begin{enumerate}
  \item the expression \(E\) is \emph{linearized} by indexing distinctively the occurrences of the symbols
        producing an expression denoted by \(E^\sharp \); as an example, if \(E={(a+b)}^*a(a+b)\),
        then \(E^\sharp={(a_1+b_2)}^*a_3(a_4+b_5)\).
        \textbf{N.B.:} Notice that the starting index of the linearization and the order involved
        do not matter. All we care about is having distinct indices.
  \item The automaton \(G_{E^\sharp}\) is then computed as before.
  \item The position automaton \(G_E\) of \(E\) is then obtained by relabelling the transitions of \(G_{E^\sharp}\)
        with unindexed symbols. More formally, the \emph{delinearization function} \(\mathrm{h}\), sending any indexed symbol
        \(a_j\) to the symbol \(a\), is applied over the transitions labels.
\end{enumerate}
The automaton \(G_E\) recognizes \(L(E)\). It is a consequence of the fact that
\begin{itemize}
  \item \(L(E^\sharp)\) is equal to \(L(G_{E^\sharp})\) by construction,
  \item \(\mathrm{h}(L(E^\sharp))\) is equal to \(L(E)\),
  \item \(\mathrm{h}(L(G_{E^\sharp}))\) is equal to \(L(G_{E})\)
\end{itemize}
where \(\mathrm{h}\) is linearly extended to sets and as a (free) monoid morphism over words.

Finally, extending the inductive computation to expressions with occurrences of \(\emptyset\),
by setting
\begin{equation*}
  \mathrm{Pos}(\emptyset) = \mathrm{First}(\emptyset) = \mathrm{Last}(\emptyset) =\mathrm{Follow}(\emptyset) = \emptyset
\end{equation*}
preserves the correction of the computation (but may lead to not accessible states
or to not coaccessible states).

\begin{example}
  Let us consider the expression \( E={(a+b)}^*a(a+b)\) and its linearized version
  \(E^\sharp={(a_1+b_2)}^*a_3(a_4+b_5)\).
  The associated position functions produced the following sets:
  \begin{align*}
    \mathrm{Null}(E^\sharp)   & = \mathrm{false},              \\
    \mathrm{Pos}(E^\sharp)    & = \{a_1, b_2, a_3, a_4, b_5\}, \\
    \mathrm{First}(E^\sharp)  & = \{a_1, b_2, a_3\},           \\
    \mathrm{Last}(E^\sharp)   & = \{a_4, b_5\},                \\
    \mathrm{Follow}(E^\sharp) & = \{
    (a_1, a_1), (a_1, b_2), (a_1, a_3),
    (b_2, a_1),                                                \\
                              & \qquad (b_2, b_2), (b_2, a_3),
    (a_3, b_4), (a_3, b_5)
    \},                                                        \\
  \end{align*}
  leading to the automaton in Figure~\ref{fig glu ex}.
\end{example}

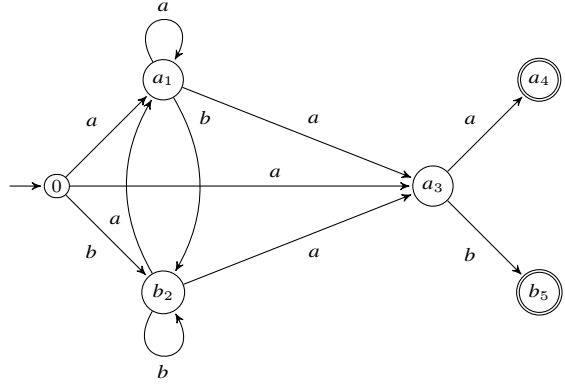
\begin{figure}[H]
  \centering
  \begin{tikzpicture}[node distance=3cm,bend angle=30]
    \node[state, rounded rectangle, initial left, initial text={}] (1)  {\(0\)} ;
    \node[state, above right of=1,node distance = 2cm] (2) {\( a_1 \)};
    \node[state, below right of=1,node distance = 2cm] (3) {\( b_2 \)};
    \node[state, right of=1,node distance = 5cm] (4) {\( a_3 \)};
    \node[state, above right of=4,node distance = 2cm, accepting] (5) {\( a_4 \)};
    \node[state, below right of=4,node distance = 2cm, accepting] (6) {\( b_5 \)};

    \path[->]
    (1) edge[ above left] node {\( a \)} (2)
    (1) edge[ below left] node {\( b \)} (3)
    (1) edge[ pos = 0.6] node {\( a \)} (4)

    (2) edge[out = 120, in = 60, loop, above] node {\( a \)} ()
    (2) edge[ pos = 0.2, bend left, above right] node {\( b \)} (3)
    (2) edge[ above right] node {\( a \)} (4)

    (3) edge[out = -120, in = -60, loop, below] node {\( b \)} ()
    (3) edge[ pos = 0.3, bend left, left] node {\( a \)} (2)
    (3) edge[ below right] node {\( a \)} (4)

    (4) edge[ above left] node {\( a \)} (5)
    (4) edge[ below left] node {\( b \)} (6)

    ;
  \end{tikzpicture}
  \caption{The Glushkov automaton of \(E\).}%
  \label{fig glu ex}
\end{figure}

\subsection{Construction for Constrained Tildes}

First, let us show that the linearization process is \emph{compatible}
with the constrained tildes.
For that purpose, let us consider the operation \(E^{\sharp, j}\) starting the linearization
at the index \(j\).
By convenience, we usually state that \(E^\sharp = E^{\sharp, 1}\).

\begin{lemma}\label{lem compat pos}
  Let \(E\) be an extended expression.
  Then
  \begin{equation*}
    L(E) = \mathrm{h}(L(E^\sharp)).
  \end{equation*}
\end{lemma}
\begin{proof}
  Let us show by induction the more general statement that
  \begin{equation*}
    L(E) = \mathrm{h}(L(E^{\sharp, j})).
  \end{equation*}
  We only exhibit the case of the sum and of a constrained tilde,
  the other case being equivalently provable.
  \begin{itemize}
    \item If \(E = F + G\), then by definition there exist two integers \((j_1, j_2)\) such that
          \(E^\sharp = F^{\sharp, j_1} + G^{\sharp, j_2}\).
          \begin{align*}
            L(F + G) & = L(F) \cup L(G)                                                     & (\textbf{Def.: language of a sum})     \\
                     & = \mathrm{h}(L(F^{\sharp, j_1})) \cup \mathrm{h}(L(G^{\sharp, j_2})) & (\textbf{Induction Hypothesis})        \\
                     & = \mathrm{h}(L(F^{\sharp, j_1}) \cup L(G^{\sharp, j_2}))             & (\textbf{linearity of \(\mathrm{h}\)}) \\
                     & = L(F^{\sharp, j_1} + G^{\sharp, j_2})                               & (\textbf{Def.: language of a sum})
          \end{align*}
    \item If \(E = \phi(E_1, \ldots, E_n)\), then by definition there exist \(n\) integers \((j_1, \ldots, j_n)\) such that
          \(E^\sharp = \phi({E_1}^{\sharp, j_1}, \ldots, {E_n}^{\sharp, j_n})\).
          \begin{align*}
            L(E)                                                    & = \bigcup_{i\mid \mathrm{eval}_i(\phi)}  i(L(E_1), \ldots, L(E_n))                                             & (\textbf{Eq.~\eqref{def language}}) \\
                                                                    & = \bigcup_{i\mid \mathrm{eval}_i(\phi)} L'_{i, 1} \cdots L'_{i, n}                                                                                   \\
                                                                    & \qquad \text{where }
            L'_{i, k} =
            \begin{cases}
              \{\varepsilon \} & \text{ if } i(k),  \\
              L(E_k)           & \text{ otherwise.}
            \end{cases}                & (\textbf{Def. of action of \(i \)})                                                                                                                                               \\
                                                                    & = \bigcup_{i\mid \mathrm{eval}_i(\phi)} L'_{i, 1} \cdots L'_{i, n}                                                                                   \\
                                                                    & \qquad \text{where }
            L'_{i, k} =
            \begin{cases}
              \{\varepsilon \}                   & \text{ if } i(k),  \\
              \mathrm{h}(L({E_k}^{\sharp, j_k})) & \text{ otherwise.}
            \end{cases} & (\textbf{Induction hypothesis})                                                                                                                                                         \\
                                                                    & = \bigcup_{i\mid \mathrm{eval}_i(\phi)} \mathrm{h}(L'_{i, 1} \cdots L'_{i, n})                                                                       \\
                                                                    & \qquad \text{where }
            L'_{i, k} =
            \begin{cases}
              \{\varepsilon \}       & \text{ if } i(k),  \\
              L({E_k}^{\sharp, j_k}) & \text{ otherwise.}
            \end{cases}             & (\textbf{\(\mathrm{h}\): monoid morphism})                                                                                                                                           \\
                                                                    & = \mathrm{h}(\bigcup_{i\mid \mathrm{eval}_i(\phi)} L'_{i, 1} \cdots L'_{i, n})                                                                       \\
                                                                    & \qquad \text{where }
            L'_{i, k} =
            \begin{cases}
              \{\varepsilon \}       & \text{ if } i(k),  \\
              L({E_k}^{\sharp, j_k}) & \text{ otherwise.}
            \end{cases}             & (\textbf{linearity of \(\mathrm{h}\)})                                                                                                                                               \\
                                                                    & = \mathrm{h}(\bigcup_{i\mid \mathrm{eval}_i(\phi)}  i(L({E_1}^{\sharp, j_1}), \ldots, L({E_n}^{\sharp, j_n}))) & (\textbf{Eq.~\eqref{def language}}) \\
                                                                    & = L(E^\sharp)
          \end{align*}
  \end{itemize}
\end{proof}

However, the \emph{linearity} of an extended expression is not sufficient anymore in the case of
constrained tildes to ensure the correction of the Glushkov construction.

\begin{example}\label{ex pos tildes}
  Let us consider the extended expression \(E_2 = \mathrm{Mirror}_2(1, 2, 3, 4)(a, b, a, b)\), its
  linearized version \({E_2}^\sharp = \mathrm{Mirror}_2(1, 2, 3, 4)(a_1, b_2, a_3, b_4)\) and
  their denoted languages
  \begin{align*}
    L(E_2)          & = \{\varepsilon, ab, ba, abab\},                 & 
    L({E_2}^\sharp) & = \{\varepsilon, a_1b_4, b_2a_3, a_1b_2a_3b_4\}.
  \end{align*}
  The symbol \(a_3\) is in the set \(\mathrm{Last}({E_2}^\sharp)\) since it ends the word
  \(b_2a_3\). However, following the Glushkov computation of the associated automaton,
  the state \(a_3\) should be final. And since the word \(a_1b_2a_3b_4\) belongs to \(L({E_2}^\sharp)\),
  its prefix should be a path in the position automaton, leading to the word \(a_1b_2a_3\) be recognized
  by the automaton, without belonging to the denoted language of \({E_2}^\sharp\).
\end{example}
The action of
the tildes restricts the following symbols of a given position w.r.t.\ to a context,
provided by the interpretations satisfying the involved Boolean formula. It can
also contextualize the \emph{finality} of a position.
Consequently, the context defined by the different interpretations satisfying the Boolean formula
of a constrained tilde needs to be explicitly considered in the states of the position automaton,
\emph{i.e.} in the definition of the function \(\mathrm{Pos}\), and therefore in all the position functions.
In order to complete this contextualization, we need to \emph{surlinearize} the positions by adding
other indices to represent the considered interpretations.
\textbf{N.B.:} As a direct consequence, the number of positions is not linear (w.r.t.\ the size
if the expression) anymore, and can be exponentially related.

\begin{example}\label{ex 2 pos tildes}
  Let us consider the Boolean formula \(\mathrm{Mirror}_2(1, 2, 3, 4)\).
  There are four interpretations satisfying it:
  \begin{itemize}
    \item \(i_1\), when the four atoms are false,
    \item \(i_2\), when \(1\) and \(4\) are true and the other atoms are false,
    \item \(i_3\), when \(2\) and \(3\) are true and the other atoms are false,
    \item \(i_4\),  when the four atoms are true.
  \end{itemize}
  Consequently, following Equation~\eqref{def language}, the expression \({E_2}^\sharp\)
  of Example~\ref{ex pos tildes} is equivalent to the expression
  \begin{equation*}
    {E'_2}^\sharp = \varepsilon + a_1b_4 + a_2b_3 + a_1b_2a_3b_4,
  \end{equation*}
  since
  \begin{equation*}
    L({E_2}^\sharp)= \bigcup_{i\in \{i_1,i_2,i_3,i_4\}} i(\{a_1\},\{b_2\},\{a_3\},\{b_4\}).
  \end{equation*}
  A surlinearization can be performed by applying the index of the interpretation over each part of
  \({E'_2}^\sharp \), leading to the expression
  \begin{equation*}
    {E_2}^{\musDoubleSharp} = \varepsilon + a_{1, 2}b_{4, 2} + a_{2,3}b_{3, 3} + a_{1, 4}b_{2, 4}a_{3, 4}b_{4, 4}.
  \end{equation*}
\end{example}
Notice that this surlinearization may happen for several nested constrained tildes.
Consequently, the surlinearization will not produce an index based on a couple of integers,
but an index based on a couple of an integer and an integer list, where each step adds an integer in the list.
As an example, the surlinearization would produce the extended expression
\begin{equation*}
  {E_2}^{\musDoubleSharp} = \varepsilon + a_{1, [2]}b_{4, [2]} + a_{2, [3]}b_{3, [3]} + a_{1, [4]}b_{2, [4]}a_{3, [4]}b_{4, [4]}
\end{equation*}
in Example~\ref{ex 2 pos tildes}.
Therefore, from now, we consider that the linearization of an expression indexes each occurrence
of an expression by a couple \((k, [])\), where \(k\) is an integer and \([]\) the empty list.

However, we will not surlinearize a linear expression before the computation of the position functions.
Instead, we surlinearize it \emph{during} the computation, whenever a constrained tilde is reached,
by considering the following computation.
This can be efficient for the implementation, as far as lazy evaluated languages are considered.

\begin{definition}
  Let \(\phi(E_1, \ldots, E_n)\) be a linearized constrained tildes expression
  and \( \{ i_1, \ldots, i_k \} \) the set of the interpretations satisfying \( \phi \).
  We denote by \(\mathrm{dev}_\phi(E_1, \ldots, E_n)\) the extended expression
  \( \sum_{j \in \{1, \ldots, k\}} \mathrm{dev}_{i_j}(E_1, \ldots, E_n) \),
  where
  \(\mathrm{dev}_{i_j}(E_1, \ldots, E_n) = \prod_{1 \leq m \leq n \mid i_j(m)} \mathrm{ind}_j(E_m)\)
  where
  \(\mathrm{ind}_j(E_m)\) is the extended expression obtained by substituting any
  occurrence of \(a_{k, [c_1, \ldots, c_l]}\) with \(a_{k, [j, c_1, \ldots, c_l]}\).
\end{definition}

\begin{example}\label{ex dev phi}
  Let us consider the expression \(E = |1 \leftrightarrow 3|(a,(|1 \rightarrow 2|(b,c)),d)\)
  and its linearized version \(E^\sharp = |1 \leftrightarrow 3|(a_{1, []},(|1 \rightarrow 2|(b_{2, []}, c_{3, []})),d_{4, []})\).
  \begin{align*}
    \mathrm{dev}_{|1 \leftrightarrow 3|}(a_1,(|1 \rightarrow 2|(b_2,c_3)),d_4) & = \varepsilon                                                                          \\
                                                                               & \qquad + (|1 \rightarrow 2|(b_{2, [2]}, c_{3, [2]}))                                   \\
                                                                               & \qquad + a_{1, [3]} \cdot d_{4, [3]}                                                   \\
                                                                               & \qquad + a_{1, [4]} \cdot (|1 \rightarrow 2|(b_{2, [4]}, c_{3, [4]})) \cdot d_{4, [4]} \\
    \mathrm{dev}_{|1 \rightarrow 2|}(b_{2, [2]}, c_{3, [2]})                   & = \varepsilon                                                                          \\
                                                                               & \qquad + b_{2, [2, 2]}                                                                 \\
                                                                               & \qquad + b_{2, [3, 2]} \cdot c_{3, [3, 2]}                                             \\
    \mathrm{dev}_{|1 \rightarrow 2|}(b_{2, [4]}, c_{3, [4]})                   & = \varepsilon                                                                          \\
                                                                               & \qquad + b_{2, [2, 4]}                                                                 \\
                                                                               & \qquad + b_{2, [3, 4]} \cdot c_{3, [3, 4]}                                             \\
  \end{align*}
\end{example}
This transformation preserves the delinearized language, since it is based on
the development associated with the language of a constrained tilde
where the only modification made is adding an index to the list,
that does not modify the action of \(\mathrm{h}\).
\begin{lemma}\label{lem lang dev phi}
  Let \(\phi(E_1, \ldots, E_n)\) be a linearized constrained tilde expression.
  Then
  \begin{equation*}
    \mathrm{h}(L(\phi({E_1}^{\sharp, j}, \ldots, {E_n}^{\sharp, j_n})))
    =
    \mathrm{h}(L(\mathrm{dev}_\phi({E_1}^{\sharp, j}, \ldots, {E_n}^{\sharp, j_n}))).
  \end{equation*}
\end{lemma}
Using this operation, the position functions are extended to constrained tildes as follows.
\begin{definition}\label{def glu fun tildes}
  Let \(\phi(E_1, \ldots, E_n)\) be a linearized constrained tilde expression.
  Then
  \begin{align*}
    \mathrm{Pos(E)}    & = \mathrm{Pos}(\mathrm{dev}_\phi(E_1, \ldots, E_n)),    & 
    \mathrm{First(E)}  & = \mathrm{First}(\mathrm{dev}_\phi(E_1, \ldots, E_n)),     \\
    \mathrm{Last(E)}   & = \mathrm{Last}(\mathrm{dev}_\phi(E_1, \ldots, E_n)),   & 
    \mathrm{Follow(E)} & = \mathrm{Follow}(\mathrm{dev}_\phi(E_1, \ldots, E_n)).
  \end{align*}
\end{definition}
The \(\mathrm{Null}\) function can be computed inductively following Corollary~\ref{cor eps in lang}.

\begin{example}
  Let us consider the expression \(E^\sharp\) of Example~\ref{ex dev phi}.
  Then:
  \begin{align*}
    \mathrm{Pos}(E^\sharp)    & =  \{(a,1,[3]),(a,1,[4]),(b,2,[2,2]),(b,2,[2,4]),(b,2,[3,2]),     \\
                              & \qquad (b,2,[3,4]),(c,3,[3,2]),(c,3,[3,4]),(d,4,[3]),(d,4,[4])\}, \\
    \mathrm{First}(E^\sharp)  & = \{(a,1,[3]),(a,1,[4]),(b,2,[2,2]),(b,2,[3,2])\},                \\
    \mathrm{Last}(E^\sharp)   & = \{(b,2,[2,2]),(c,3,[3,2]),(d,4,[3]),(d,4,[4])\}                 \\
    \mathrm{Follow}(E^\sharp) & =
    \{((a,1,[3]),(d,4,[3])) \}
    \cup \{(a,1,4)\} \times \{(b,2,2,4),(b,2,3,4),(d,4,4)\}                                       \\
                              & \qquad \cup \{
    ((b,2,[2,4]),(d,4,[4])),
    ((b,2,[3,2]), (c,3,[3,2])),                                                                   \\
                              & \qquad \quad ((b,2,[3,4]),(c,3,[3,4])),
    ((c,3,[3,4]), (d,4,[4]))
    \}.
  \end{align*}
\end{example}
The classical construction, detailed in Subsection~\ref{subsec clas cons}, can then be applied to any extended expression \(E\):
\begin{enumerate}
  \item the expression \(E\) is \emph{linearized} by indexing distinctively the occurrences of the symbols
        with a couple made of an integer and the empty list,
        producing an expression denoted by \(E^\sharp \).
  \item The position functions are computed using Definition~\ref{def glu fun tildes}.
  \item The automaton \(G_{E^\sharp}\) is then computed as before following Definition~\ref{def glu aut};
  \item The position automaton \(G_E\) of \(E\) is then obtained by relabelling the transitions of \(G_{E^\sharp}\)
        with unindexed symbols.
\end{enumerate}

\begin{example}
  Let us consider the expression \(E\) of Example~\ref{ex dev phi}.
  The Glushkov automaton of \(E\) is presented in Figure~\ref{fig glu cons tildes}.
\end{example}

\begin{figure}[H]
  \centering
  \begin{tikzpicture}[node distance=3cm,bend angle=30]
    \node[state, rounded rectangle, initial left, initial text={}, accepting] (1)  {\(0\)} ;

    \node[state, rounded rectangle, above of=1,node distance = 2cm, accepting] (2) {\( b_{(2, [2,2])} \)};

    \node[state, rounded rectangle, above right of=1,node distance = 2cm] (3) {\( a_{(1, [3])} \)};
    \node[state, rounded rectangle, right of=3,node distance = 2cm, accepting] (4) {\( d_{(4, [3])} \)};

    \node[state, rounded rectangle, below of=1,node distance = 2cm] (5) {\( b_{(2, [3,2])} \)};
    \node[state, rounded rectangle, right of=5,node distance = 2cm, accepting] (6) {\( c_{(3, [3,2])} \)};

    \node[state, rounded rectangle, right of=1,node distance = 2cm] (7) {\( a_{(1, [4])} \)};
    \node[state, rounded rectangle, right of=7,node distance = 3cm] (8) {\( b_{(2, [2, 4])} \)};
    \node[state, rounded rectangle, right of=8,node distance = 3cm, accepting] (9) {\( d_{(4, [4])} \)};

    \node[state, rounded rectangle, below right of=7,node distance = 2cm] (10) {\( b_{(2, [3, 4])} \)};
    \node[state, rounded rectangle, below left of=9,node distance = 2cm] (11) {\( c_{(3, [3, 4])} \)};

    \path[->]
    (1) edge[ left] node {\( b \)} (2)

    (1) edge[ above left] node {\( a \)} (3)
    (3) edge[ above] node {\( d \)} (4)

    (1) edge[ left] node {\( b \)} (5)
    (5) edge[ above] node {\( c \)} (6)

    (1) edge[ above] node {\( a \)} (7)
    (7) edge[ above] node {\( b \)} (8)
    (8) edge[ above] node {\( d \)} (9)

    (7) edge[ below left] node {\( b \)} (10)
    (10) edge[ below] node {\( c \)} (11)
    (11) edge[ below right] node {\( d \)} (9)

    (7) edge[ above, bend left] node {\( d \)} (9)





    ;
  \end{tikzpicture}
  \caption{The Glushkov automaton of \(E\).}%
  \label{fig glu cons tildes}
\end{figure}
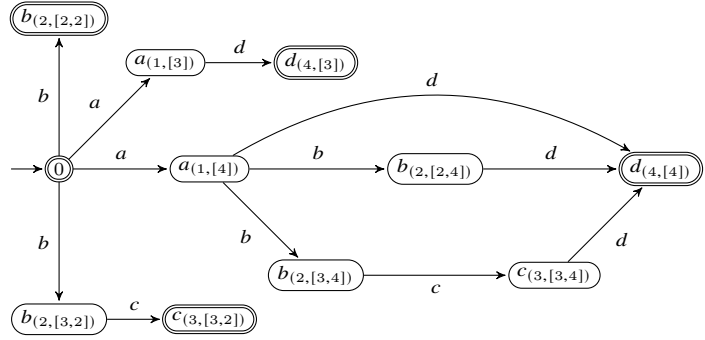

\subsection{Correction of the Construction}

In order to prove the correction of the construction, let us consider the \emph{surlinearized expression}
\(E^{\musDoubleSharp} = \mathrm{dev}(E^\sharp)\) of an extended expression \(E\) computed inductively as follows:
\begin{gather*}
  \begin{aligned}
    \mathrm{dev}(a_{j, js}) & = a_{j, js},                         & \mathrm{dev}(\emptyset) & = \emptyset,                             & \mathrm{dev}(\varepsilon) & = \varepsilon,       \\
    \mathrm{dev}(F+G)       & = \mathrm{dev}(F) + \mathrm{dev}(G), & \mathrm{dev}(F\cdot G)  & = \mathrm{dev}(F) \cdot \mathrm{dev}(G), & \mathrm{dev}(F^*)         & = \mathrm{dev}(F)^*,
  \end{aligned}\\
  \begin{aligned}
    \mathrm{dev}(\phi(E_1, \ldots, E_n)) & = \mathrm{dev}(\mathrm{dev}_\phi(E_1, \ldots, E_n)).
  \end{aligned}
\end{gather*}
Notice that by construction, \(E^{\musDoubleSharp}\) is a regular expression.

\begin{example}
  Let us consider the expression \(E = |1 \leftrightarrow 3|(a,(|1 \rightarrow 2|(b,c)),d)\)
  and its linearized version \(E^\sharp = |1 \leftrightarrow 3|(a_{1, []},(|1 \rightarrow 2|(b_{2, []}, c_{3, []})),d_{4, []})\)
  of Example~\ref{ex dev phi}. Then:
  \begin{align*}
    E^{\musDoubleSharp} & = \mathrm{dev}(\varepsilon                                                              \\
                        & \qquad + (|1 \rightarrow 2|(b_{2, [2]}, c_{3, [2]}))                                    \\
                        & \qquad + a_{1, [3]} \cdot d_{4, [3]}                                                    \\
                        & \qquad + a_{1, [4]} \cdot (|1 \rightarrow 2|(b_{2, [4]}, c_{3, [4]})) \cdot d_{4, [4]}) \\
                        & = \varepsilon                                                                           \\
                        & \qquad + (\varepsilon                                                                   \\
                        & \qquad \qquad + b_{2, [2, 2]}                                                           \\
                        & \qquad \qquad + b_{2, [3, 2]} \cdot c_{3, [3, 2]})                                      \\
                        & \qquad + a_{1, [3]} \cdot d_{4, [3]}                                                    \\
                        & \qquad + a_{1, [4]} \cdot (                                                             \\
                        & \qquad \qquad \varepsilon                                                               \\
                        & \qquad \qquad + b_{2, [2, 4]}                                                           \\
                        & \qquad \qquad + b_{2, [3, 4]} \cdot c_{3, [3, 4]} )\cdot d_{4, [4]}
  \end{align*}
\end{example}
By a trivial induction over the structure of extended expressions, it can be shown that the surlinearization preserves
the position functions.
\begin{proposition}
  Let \(E\) be an extended expression.
  Then:
  \begin{gather*}
    \begin{aligned}
      \mathrm{Pos}(E^{\musDoubleSharp})   & = \mathrm{Pos}(E^\sharp),   & 
      \mathrm{First}(E^{\musDoubleSharp}) & = \mathrm{First}(E^\sharp), & 
      \mathrm{Last}(E^{\musDoubleSharp})  & = \mathrm{Last}(E^\sharp),
    \end{aligned}\\
    \begin{aligned}
      \mathrm{Follow}(E^{\musDoubleSharp}) & = \mathrm{Follow}(E^\sharp), & 
      \mathrm{Null}(E^{\musDoubleSharp})   & = \mathrm{Null}(E^\sharp).
    \end{aligned}
  \end{gather*}
\end{proposition}
\begin{proof}
  Let us prove the case of the \(\mathrm{First}\) function by induction over the structure of the expressions
  and by recurrence over the number of constrained tildes, the other cases being equivalently provable.
  As in the proof of Lemma~\ref{lem compat pos}, let us show the more general result
  that \(\mathrm{First}(\mathrm{dev}(E^{\sharp, j})) = \mathrm{First}(E^{\sharp, j})\).
  Once again, we restrict the proof to two cases, the other ones being equivalently provable.
  \begin{itemize}
    \item If \(E^{\sharp, j} = F^{\sharp, j} + G^{\sharp, j_2}\), then
          \begin{align*}
            \mathrm{First}(E^{\sharp, j}) & = \mathrm{First}(F^{\sharp, j}) \cup \mathrm{First}(G^{\sharp, j_2})                             & (\textbf{Def.: \(\mathrm{First}\)}) \\
                                          & = \mathrm{First}(\mathrm{dev}(F^{\sharp, j})) \cup \mathrm{First}(\mathrm{dev}(G^{\sharp, j_2})) & (\textbf{Induction Hypothesis})     \\
                                          & = \mathrm{First}(\mathrm{dev}(F^{\sharp, j}) + \mathrm{dev}(G^{\sharp, j_2}))                    & (\textbf{Def.: \(\mathrm{First}\)}) \\
                                          & = \mathrm{First}(\mathrm{dev}(E^{\sharp, j}))                                                    & (\textbf{Def.: \(\mathrm{dev}\)})
          \end{align*}
    \item If \(E^{\sharp, j} = \phi({E_1}^{\sharp, j}, \ldots, {E_n}^{\sharp, j_n})\), then
          \begin{align*}
            \mathrm{First}(E^{\sharp, j}) & = \mathrm{First}(\mathrm{dev}_\phi({E_1}^{\sharp, j}, \ldots, {E_n}^{\sharp, j_n}))               & (\textbf{Def.: \(\mathrm{First}\)}) \\
                                          & = \mathrm{First}(\mathrm{dev}(\mathrm{dev}_\phi({E_1}^{\sharp, j}, \ldots, {E_n}^{\sharp, j_n}))) & (\textbf{Induction Hypothesis})     \\
                                          & = \mathrm{First}(\mathrm{dev}(E^{\sharp, j}))                                                     & (\textbf{Def.: \(\mathrm{dev}\)})
          \end{align*}\\
  \end{itemize}
\end{proof}
As a direct consequence, since \(E^{\musDoubleSharp}\) is a regular expression, the classical
Glushkov construction for regular expressions for \(E^{\musDoubleSharp}\) coincides
with the Glushkov automaton of the extended expression \(E^\sharp \).
\begin{corollary}
  Let \(E\) be an extended expression.
  Then \(G_{E^\sharp}\) and \(G_{E^{\musDoubleSharp}}\) are equal.
\end{corollary}
This allows us to characterize the language recognized by \(G_{E^\sharp}\).
\begin{corollary}
  Let \(E\) be an extended expression.
  Then \(L(G_{E^\sharp})\) is equal to \(L(E^{\musDoubleSharp})\).
\end{corollary}

Another important property of surlinearization is that after delinearization of the languages denoted by \(E^\sharp \) or \(E^{\musDoubleSharp}\), using the function \(\mathrm{h}\),
we obtain the language of the starting expression.
\begin{proposition}
  Let \(E\) be an extended expression.
  Then
  \(\mathrm{h}(L(E^\sharp)) = \mathrm{h}(L(E^{\musDoubleSharp}))\).
\end{proposition}
\begin{proof}
  We proceed by induction over the structure of the expressions
  and by recurrence over the number of constrained tildes.
  As in the proof of Lemma~\ref{lem compat pos}, let us show the more general result
  that \(\mathrm{h}(L(\mathrm{dev}(E^{\sharp, j}))) = \mathrm{h}(L(E^{\sharp, j}))\).
  Once again, we restrict the proof to two cases, the other ones being equivalently provable.
  \begin{itemize}
    \item If \(E^{\sharp, j} = F^{\sharp, j} + G^{\sharp, j_2}\), then
          \begin{align*}
            \mathrm{h}(L(E^{\sharp, j})) & = \mathrm{h}(L(F^{\sharp, j}) \cup L(G^{\sharp, j_2}))                                         & (\textbf{Def.: language of sum})       \\
                                         & = \mathrm{h}(L(F^{\sharp, j})) \cup \mathrm{h}(L(G^{\sharp, j_2}))                             & (\textbf{linearity of \(\mathrm{h}\)}) \\
                                         & = \mathrm{h}(L(\mathrm{dev}(F^{\sharp, j}))) \cup \mathrm{h}(L(\mathrm{dev}(G^{\sharp, j_2}))) & (\textbf{Induction Hypothesis})        \\
                                         & = \mathrm{h}(L(\mathrm{dev}(F^{\sharp, j})) \cup L(\mathrm{dev}(G^{\sharp, j_2})))             & (\textbf{linearity of \(\mathrm{h}\)}) \\
                                         & = \mathrm{h}(L(\mathrm{dev}(F^{\sharp, j}) + \mathrm{dev}(G^{\sharp, j_2})))                   & (\textbf{Def.: language of sum})       \\
                                         & = \mathrm{h}(L(\mathrm{dev}(E^{\sharp, j})))                                                   & (\textbf{Def.: \(\mathrm{dev}\)})
          \end{align*}
    \item If \(E^{\sharp, j} = \phi({E_1}^{\sharp, j}, \ldots, {E_n}^{\sharp, j_n})\), then
          \begin{align*}
            \mathrm{h}(L(E^{\sharp, j})) & = \mathrm{h}(L(\phi({E_1}^{\sharp, j}, \ldots, {E_n}^{\sharp, j_n})))                            & (\textbf{Def.: language of a tilde})    \\
                                         & = \mathrm{h}(L(\mathrm{dev}_\phi({E_1}^{\sharp, j}, \ldots, {E_n}^{\sharp, j_n})))               & (\textbf{Lemma~\ref{lem lang dev phi}}) \\
                                         & = \mathrm{h}(L(\mathrm{dev}(\mathrm{dev}_\phi({E_1}^{\sharp, j}, \ldots, {E_n}^{\sharp, j_n})))) & (\textbf{Induction Hypothesis})         \\
                                         & = \mathrm{h}(L(\mathrm{dev}(E^{\sharp, j})))                                                     & (\textbf{Def.: \(\mathrm{dev}\)})
          \end{align*}\\
  \end{itemize}
\end{proof}
In combination with Lemma~\ref{lem compat pos}, we obtain the following result.
\begin{corollary}
  Let \(E\) be an extended expression.
  Then
  \begin{equation*}
    L(E) = \mathrm{h}(L(E^\sharp)) = \mathrm{h}(L(E^{\musDoubleSharp})).
  \end{equation*}
\end{corollary}

Finally, since after relabelling the transitions of \(G_{E^{\musDoubleSharp}}\) we obtain
an automaton recognizing \(\mathrm{h}(L(E^{\musDoubleSharp}))\), we can determine
the language recognized by \(G_E\).
\begin{theorem}
  Let \(E\) be an extended expression.
  Then \(L(G_E) = L(E)\).
\end{theorem}

\section{Haskell Implementation}\label{sec haskell}

The computations of partial derivatives, of derived term automata and of Glushkov automata
have been implemented in Haskell and is publicly available
on GitHub~\cite{LM23}.
Constrained tildes are implemented using dependently typed programming:
a Boolean formula encoding a constrained tildes operator uses an alphabet the size of which cannot be greater than the length of the list
of expressions the formula is applied on.
Derived term and Glushkov automata can be graphically represented using Dot and Graphviz, and converted in PNG\@.
A parser from string declaration is also included.
A web interface is also available, constructed through reactive functional programming
using Reflex and Reflex-Dom~\cite{reflex, reflex-dom}.
Notice that for implementation consideration, the Boolean formulae atoms start at position \( 0 \).
\begin{figure}[H]
  \begin{center}
    \includegraphics[scale=0.3]{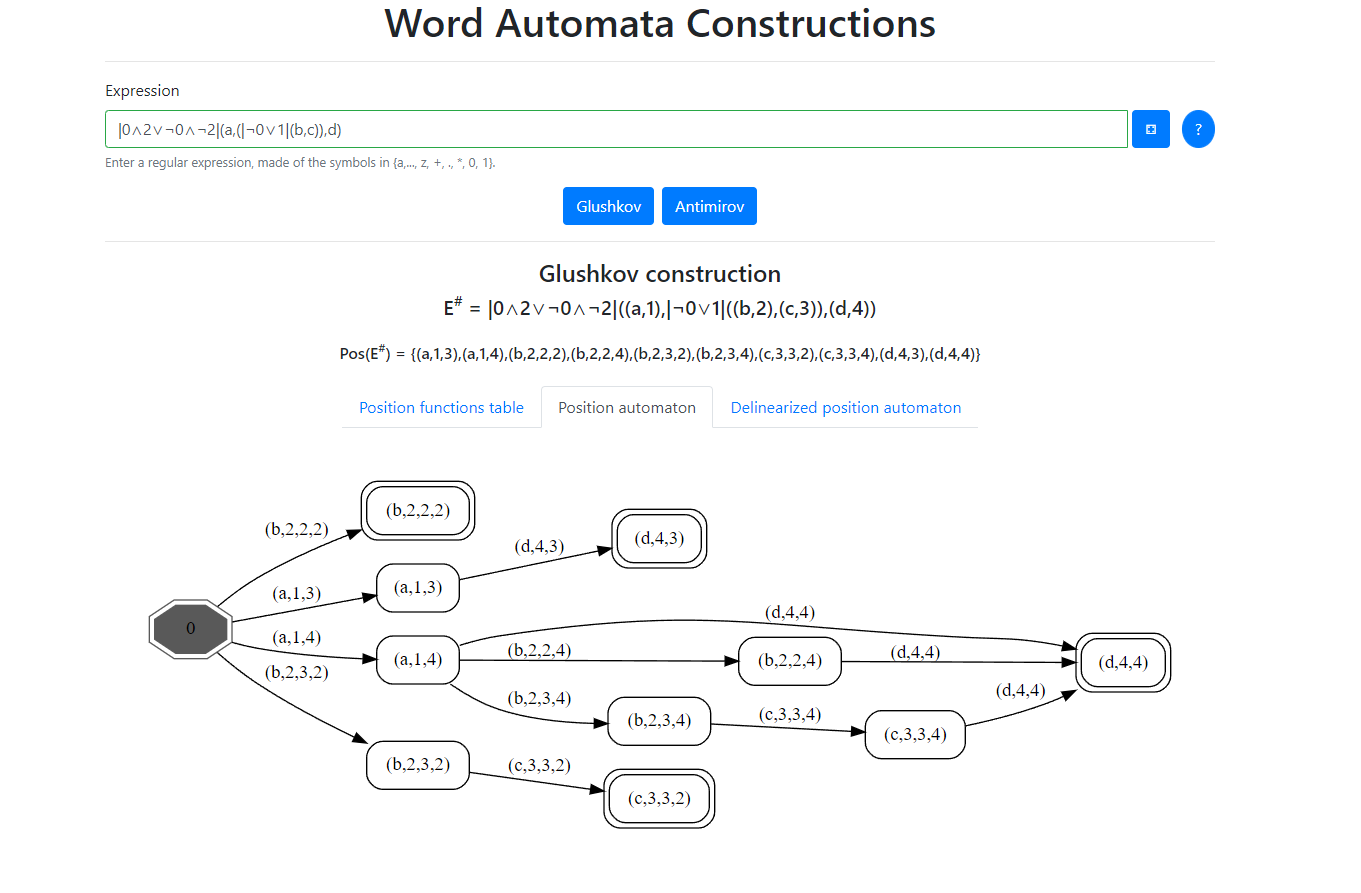}
  \end{center}
  \caption{The Web Interface.}
\end{figure}

\section{Conclusion and Perspectives}

In this paper, we have extended (disjunctive) multi-tildes operators to any Boolean combinations
of these tildes, the constrained multi-tildes, and defined their denoted languages.
We have shown that the action of these operators preserves regularity, that they may lead
to exponentially smaller expressions and how to solve the membership test by defining the partial
derivatives, the (finite) derived term automaton and the position automaton.

The next step of our plan is to study the conversion of an automaton into an equivalent expression,
by first characterizing the structures of derived term automaton and Glushkov automaton, like it was previously done in the classical
regular case~\cite{CZ00,LS05}.

\bibliographystyle{ws-ijfcs}
\bibliography{biblio}

\end{document}